\documentclass[10pt,draftcls,oneside,onecolumn]{IEEEtran}

\usepackage{graphicx,cite}
\usepackage{amsmath}
\usepackage{amssymb}
\usepackage{amsthm}
\usepackage{subfigure}
\usepackage{psfrag}
\usepackage{dsfont}
\usepackage{algorithm}
\usepackage{algpseudocode}
\usepackage[T1]{fontenc}         
\newcommand*{\changefont}[3]{%
\fontfamily{#1}\fontseries{#2}\fontshape{#3}\selectfont}

\newtheorem{theorem}{Theorem}
\numberwithin{theorem}{section}
\newtheorem{defin}[theorem]{Definition}
\newtheorem{example}[theorem]{Example}
\newtheorem{lemma}[theorem]{Lemma}
\newtheorem{remark}[theorem]{Remark}

\begin{document}

\title{
Various Views on the Trapdoor Channel and an Upper Bound on its Capacity\thanks{\noindent 
Tobias Lutz is with the
Lehrstuhl f\"ur Nachrichtentechnik,
Technische Universit\"at M\"unchen, D-80290 M\"unchen, Germany
(e-mail: tobi.lutz@tum.de).
}}

\author{Tobias Lutz}

\maketitle

\begin{abstract}
Two novel views are presented on the trapdoor channel. First, by deriving the underlying iterated function system (IFS), it is shown that the trapdoor channel with input blocks of length~$n$ can be regarded as the $n$th element of a sequence of shapes approximating a fractal. Second, an algorithm is presented that fully characterizes the trapdoor channel and resembles the recursion of generating all permutations of a given string. Subsequently, the problem of maximizing a $n$-letter mutual information is considered. It is shown that~$\frac{1}{2}\log_2\left(\frac{5}{2}\right)\approx 0.6610$~bits per use is an upper bound on the capacity of the trapdoor channel. This upper bound, which is the tightest upper bound known, proves that feedback increases the capacity.
\end{abstract}

\begin{IEEEkeywords} 
Trapdoor channel, Lagrange multipliers, convex optimization, iterated function systems, fractals, channels with memory, recursions, permutations.
\end{IEEEkeywords}

\section{Introduction}
\IEEEPARstart{T}he trapdoor channel was introduced by David Blackwell in 1961~\cite{Bla61} and is used by Robert Ash both as a book cover and as an introductory example for channels with memory~\cite{Ash65}. The mapping of channel inputs to channel outputs can be described as follows. Consider a box that contains a ball that is labeled $s_0 \in \{0,1\}$, where the index~$0$ refers to time~$0$. Both the sender and the receiver know the initial ball. In time slot~$1$, the sender places a new ball labeled $x_1 \in \{0,1\}$ in the box. In the same time slot, the receiver chooses one of the two balls $s_0$ or $x_1$ at random while the other ball remains in the box. The chosen ball is interpreted as channel output $y_1$ at time $t=1$ while the remaining ball becomes the channel state $s_1$. The same procedure is applied in every future channel use. In time slot~$2$, for instance, the sender places a new ball $x_2\in \{0,1\}$ in the box and the corresponding channel output $y_2$ is either $x_2$ or $s_1$. The transmission process is visualized in Fig.~\ref{fig:trapdoor}. Fig.~\ref{fig:trapdoor1} shows the trapdoor channel at time~$t$ when the sender places ball~$x_t$ in the box. In the same time slot, the receiver chooses randomly ball $s_{t-1}$ as channel output. Consequently, the upcoming channel state~$s_t$ becomes $x_t$ (see Fig.~\ref{fig:trapdoor2}). At time~$t+1$ the sender places a new ball~$x_{t+1}$ in the box and the receiver draws $y_{t+1}$ from~$s_t$ and~$x_{t+1}$. Table~\ref{tab:transitions} depicts the probability of an output $y_t$ given an input $x_t$ and state $s_{t-1}$. 
\begin{figure}[ht]
\centering
\psfrag{0}{$0$}
\psfrag{1}{$1$}
\psfrag{xt}{$x_{t}$}
\psfrag{xt1}{$x_{t+1}$}
\psfrag{xt2}{$x_{t+2}$}
\psfrag{xt3}{$x_{t+3}$}
\psfrag{st1}{$s_{t-1}$}
\psfrag{st=xt}{$s_t = x_{t}$}
\psfrag{yt1}{$y_{t-1}$}
\psfrag{yt2}{$y_{t-2}$}
\psfrag{yt3}{$y_{t-3}$}
\psfrag{yt=st1}{$y_t = s_{t-1}$}
\subfigure[The trapdoor channel at time~$t$.]{
   \includegraphics[scale =0.85] {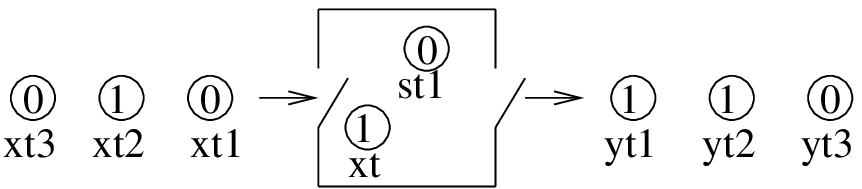}
   \label{fig:trapdoor1}
 }

 \subfigure[The trapdoor channel at time~$t+1$.]{
   \includegraphics[scale =0.85] {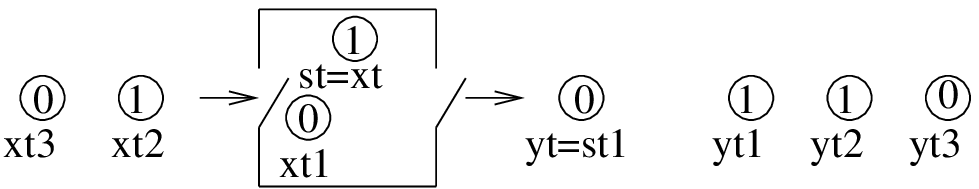}
   \label{fig:trapdoor2}
 }
\caption{At time $t$ the sender places a new ball $x_t$ in the box. The corresponding channel output~$y_t$ is $s_{t-1}$ and the next state~$s_t$ becomes~$x_t$.}
\label{fig:trapdoor}
\end{figure}

Despite the simplicity of the trapdoor channel, the derivation of its capacity seems challenging and is still an open problem. One feature that makes the problem cumbersome is that the distribution of the output symbols may depend on events happening arbitrarily far back in the past since each ball has a positive probability to remain in the channel over any finite number of channel uses. Instead of maximizing $I(X;Y)$ one rather has to consider the multi-letter mutual information, i.e., $\limsup_{n\rightarrow \infty} I(X^n;Y^n)$. 

\begin{table}[ht]
\centering
\caption{Transition probabilities of the trapdoor channel}
\label{table:achievable}
\begin{tabular}{c|c|c|c}
 $x_t$ & $s_{t-1}$ & $p(y_t = 0|x_t,s_{t-1})$ & $p(y_t = 1|x_t,s_{t-1})$ \\ \hline
0&0&1&0 \\
0&1&0.5&0.5 \\
1&0&0.5&0.5 \\
1&1&0& 1
\end{tabular}
\label{tab:transitions}
\end{table}

Let $P_{n|s_0}$ denote the matrix of conditional probabilities of output sequences of length~$n$ given input sequences of length~$n$ where the initial state equals $s_0$. The following ordering of the entries of $P_{n|s_0}$ is assumed. Row indices represent input sequences and column indices represent output sequences. To be more precise, the entry~$\begin{bmatrix}P_{n|s_0}\end{bmatrix}_{i,j}$ is the conditional probability of the binary output sequence corresponding to the integer $j-1$ given the binary input sequence corresponding the the integer $i-1$, $1 \leq i,j \leq 2^n$. For instance, if $n=3$ then $\begin{bmatrix}P_{3|s_0}\end{bmatrix}_{5,3}$ denotes the conditional probability that the channel input $x_1 x_2 x_3 = 100$ will be mapped to the channel output~$y_1 y_2 y_3 = 010$. It was shown in~\cite{KoMo02} that the conditional probability matrices $P_{n|s_0}$ satisfy the recursion laws
\begin{align}
  \label{Pn+1|0}
  P_{n+1|0} &= \begin{bmatrix} P_{n|0} & 0 \\ \frac{1}{2}P_{n|1} & \frac{1}{2}P_{n|0}\\\end{bmatrix}\\
  \label{Pn+1|1}
  P_{n+1|1} &= \begin{bmatrix} \frac{1}{2}P_{n|1} & \frac{1}{2}P_{n|0} \\ 0 & P_{n|1}\\\end{bmatrix},
\end{align}
where the initial matrices are given by $P_{0|0} = P_{0|1} = [1]$. A quick inspection of $P_{2|0}$ and $P_{2|1}$ reveals that the inputs $00$ and $11$ are mapped to disjoint outputs. Hence, a rate of $0.5$~bits per use~(b/u) is achievable from the sender to the receiver. It was shown in~\cite{AhKa87} that~$0.5$~b/u is indeed the zero-error capacity of the trapdoor channel. 

Permuter et al.~\cite{PeHaCuRo08} considered the trapdoor channel under the additional assumption of having a unit delay feedback link available from the receiver to the sender. The sender is able to determine the state of the channel in each time slot. They established that the capacity of the trapdoor channel with feedback is equal to the logarithm of the golden ratio. One can already deduce from this quantity that the achievability scheme involves a constrained coding scheme in which certain sub-blocks are forbidden.

In this paper, we propose two different views on the trapdoor channel. Based on the underlying stochastic matrices~(\ref{Pn+1|0}) and~(\ref{Pn+1|1}), the trapdoor channel can be described geometrically as a fractal or algorithmically as a recursive procedure. We then consider the problem of maximizing the $n$-letter mutual information of the trapdoor channel for any $n \in \mathbb{N}$. We relax the problem by permitting distributions that are not probability distributions. The resulting optimization problem is convex but the feasible set is larger than the probability simplex. Using the method of Lagrange multipliers via a theorem presented in~\cite{Ash65}, we show that~$\frac{1}{2}\log_2\left(\frac{5}{2}\right)\approx 0.6610$~b/u is an upper bound on the capacity of the trapdoor channel. Specifically, the same absolute maximum  $\frac{1}{2}\log_2\left(\frac{5}{2}\right)\approx 0.6610$~b/u results for all trapdoor channels which process input blocks of even length $n$. And the sequence of absolute maxima corresponding to trapdoor channels which process inputs of odd lengths converges to $\frac{1}{2}\log_2\left(\frac{5}{2}\right)$~b/u from below as the block length increases. Unfortunately, the absolute maxima of our relaxed optimization are attained outside the probability simplex, otherwise we would have established the capacity. Nevertheless,  $\frac{1}{2}\log_2\left(\frac{5}{2}\right)\approx 0.6610$~b/u is, to the best of our knowledge, the tightest capacity upper. Moreover, this bound is less than the feedback capacity of the trapdoor channel. 

The organization of this paper is as follows. Section~\ref{sec:fractal} interprets the trapdoor channel as a fractal and derives the underlying iterated function system (IFS). Section~\ref{sec:algorithm} introduces a recursive algorithm which fully characterizes the trapdoor channel. Comments on the permuting nature of the trapdoor channel are provided. Section~\ref{sec:Lagrange} presents a solution to the optimization problem outlined above and derives various recursions. The paper concludes with Section~\ref{sec:concl}.

\subsection{Notation}
The symbols $\mathbb{N}_0$ and $\mathbb{N}$ refer to the natural numbers with and without~$0$, respectively. The canonical basis vectors of $\mathbb{R}^3$ are denoted by $e_x$, $e_y$ and $e_z$. They are assumed to be row vectors. The $n$-fold composition of a function, say~$\Phi$, is denoted as $\Phi^{\circ n}$. The input corresponding to the $i$th row of~$P_{n|s_0}$ is denoted as $x^n_i$. The input corresponding to the $i$th row of~$P_{n|s_0}$ is denoted as $x^n_i$. Further, $I_n$ denotes the $2^n\times2^n$ identity matrix, $\tilde{I}_n$ is a $2^n\times2^n$ matrix whose secondary diagonal entries are all equal to~$1$ while the remaining entries are all equal to~$0$, and $1_n$ denotes a column vector of length~$2^n$ consisting only of ones. The vector~$1_n^T$ is the transpose of $1_n$. For the sake of readability we use $\exp_2(\cdot)$ instead of $2^{(\cdot)}$. If the logarithm $\log_2(\cdot)$ or the exponential function $\exp_2(\cdot)$ is applied to a vector or a matrix, we mean that $\log_2(\cdot)$ or $\exp_2(\cdot)$ of each element of the vector or matrix is taken. Finally, the symbol $\circ$ refers to the Hadarmard product, i.e., the entrywise product of two matrices.

\section{The Trapdoor Channel and Fractal Geometry}
\label{sec:fractal}
\subsection{Prerequisites}
We briefly introduce the idea of \textit{iterated function systems} and \textit{fractals}. For a comprehensive introduction to the subject, see for instance~\cite{Bar88}. In a nutshell, a fractal is a geometric pattern which exhibits self-similarity at every scale. A systematic way for generating a fractal starts with a complete metric space $(M,d)$. The space to which the fractal belongs is, however, not $M$ but the space of non-empty compact subsets of $M$, denoted as~$\mathcal{H}(M)$. A suitable choice for a metric for $\mathcal{H}(M)$ is the Hausdorff distance $h_d(A,B):= \max\{d(A,B), d(B,A)\}$ where $d(A,B):= \max_{x \in A}\min_{y \in B}d(x,y)$, $A,B \in \mathcal{H}(M)$ and analogously for $d(B,A)$. It is then guaranteed that $(\mathcal{H}(M),h_d)$ is a complete metric space and that every contraction mapping\footnote{Let $(M,d)$ be a metric space. Recall that a mapping $\varphi:M \rightarrow M$ is a \textit{contraction} if there exists a $0 < s < 1$ such that $d\left(\varphi(x),\varphi(y)\right)\leq s\cdot d(x,y)$ for all $x,y \in M$.} $\varphi:M\rightarrow M$ on $(M,d)$ becomes a contraction mapping $\varphi:\mathcal{H}(M)\rightarrow \mathcal{H}(M)$ on $(\mathcal{H}(M),h_d)$ defined by $\varphi(A) = \{\varphi(x):x\in A\}$ for all $A \in \mathcal{H}(M)$. 

The following definition and theorem provides a method for generating fractals.
\begin{defin}{\cite[Chapter 3.7]{Bar88}}
A \textnormal{hyperbolic iterated function system (IFS)} consists of a complete metric space $(M,d)$ together with a finite set of contraction mappings $\varphi_n:M \rightarrow M$, with respective contractivity factors $s_n$ for $n = 1, 2,\dots, N$. The notation for the IFS is $\{M; \varphi_n\, n=1,2,\dots,N\}$ and its contractivity factor is $s = \max\{s_n : n = 1,2,\dots,N\}$.
\label{def:IFS}
\end{defin}
The fixed point of a hyperbolic IFS, also called the \textit{attractor} or \textit{self-similar set} of the IFS, is a (deterministic) fractal and results from iterating the IFS with respect to any $A \in \mathcal{H}(M)$. This is the content of the following theorem.
\begin{theorem}{\cite[Chapter 3.7]{Bar88}}
Let $\{M; \varphi_n\, n=1,2,\dots,N\}$ be an iterated function system with contractivity factor $s$. Then the transformation $\Phi:\mathcal{H}(M)\rightarrow \mathcal{H}(M)$ defined by
\begin{equation}
    \Phi(A) = \bigcup_{n=1}^N \varphi_n(A)
\label{thm:Phi}
\end{equation}
for all $A \in \mathcal{H}(M)$, is a contraction mapping on the complete metric space $(\mathcal{H}(M),h_d)$ with contractivity factor~$s$. Its unique fixed point, $A^\star \in \mathcal{H}(M)$, obeys
\[
    A^\star = \Phi(A^\star) = \bigcup_{n=1}^N\varphi_n(A^\star),
\]
and is given by $A^\star = \lim_{k \rightarrow \infty}\Phi^{\circ k}(A)$ for any $A \in \mathcal{H}(M)$.
\label{thm:IFS}
\end{theorem}
Many well-known fractals, e.g., the \textit{Koch snowflake}, the \textit{Cantor set}, the \textit{Mandelbrot set}, etc., can be generated using Definition~\ref{def:IFS} and Theorem~\ref{thm:IFS}. Indeed, a segment of the Mandelbrot set is shown on the cover of the book by Cover and Thomas~\cite{CovTho91}. Another famous representative, the \textit{Sierpinski triangle}, is introduced in the following example. We will later see that this fractal is related to the trapdoor channel. 
\begin{example}
\emph{(Sierpinski triangle)} 
Consider the IFS
\begin{equation}
     \left\{[0,1]^2;\varphi_1(x,y) = \left(\frac{x+1}{2},\frac{y}{2}\right),\varphi_2(x,y) = \left(\frac{x}{2},\frac{y+1}{2}\right),\varphi_3(x,y) = \left(\frac{x}{2},\frac{y}{2}\right)\right\}.
     \label{ex:IFS}
\end{equation}
The affine transformations $\varphi_n$, $n=1,2,3$, scale any $A \in \mathcal{H}([0,1]^2)$ by a factor of $0.5$. Additionally, $\varphi_1$ and~$\varphi_2$ introduce translations by $0.5$ into the $x$- and $y$-direction, respectively. The Sierpinski triangle is approximated arbitrarily close by iterating $\Phi(A)$ for any $A \in \mathcal{H}([0,1]^2)$. Fig.~\ref{fig:sierp} shows the result after performing five iterations of~(\ref{ex:IFS}). The initial shape $A$ in Fig.~\ref{fig:sierp1} is a triangle with corner points  $(0,0), (1,0), (0,1)$ and in Fig.~\ref{fig:sierp2} a triangle with corner points $(0,0), (1,1), (1,0)$. As one performs more iterations, both sets converge to the same set $A^\star$.
\end{example}
\begin{figure}[ht]
\centering
\subfigure[The initial shape is a triangle with corner points $(0,0), (1,0), (0,1)$.]{
   \includegraphics[scale =0.6] {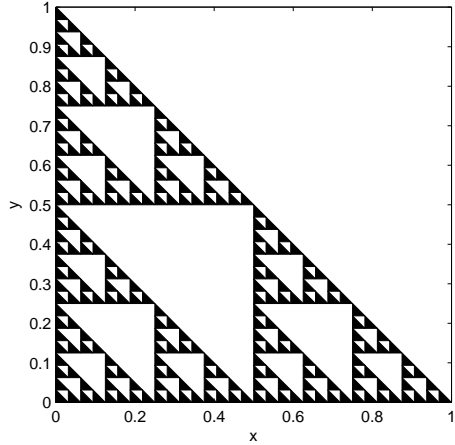}
   \label{fig:sierp1}
 }

 \subfigure[The initial shape is a triangle with corner points $(0,0), (1,1), (1,0)$.]{
   \includegraphics[scale =0.6] {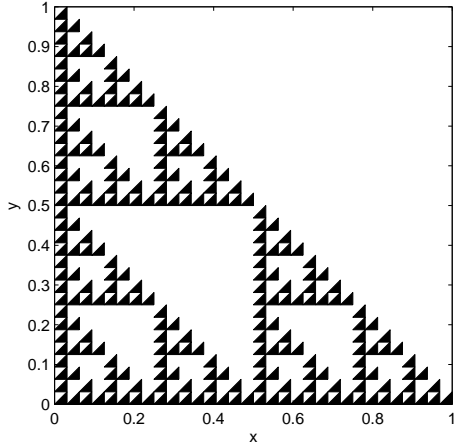}
   \label{fig:sierp2}
 }
\caption{Sierpinski triangle after four iterations of the underlying IFS with two different initial shapes.}
\label{fig:sierp}
\end{figure}

\subsection{The Trapdoor Channel as a Fractal}
In this section, we derive a hyperbolic IFS for the trapdoor channel. Instead of working with $P_{n|s_0}$ we take a geometric approach, i.e., $P_{n|s_0}$ will be mapped to the unit cube $[0,1]^3\subset \mathbb{R}^3$. 
\begin{defin}
Let $\mathcal{M}$ denote the set $\left\{P_{n|s_0}:n \in \mathbb{N}_0, s_0 = 0,1\right\}$ of trapdoor channel matrices. The function $\rho^{(n)}:\mathcal{M}\rightarrow [0,1]^3$ represents each $P_{n|s_0}$ as a shape in $[0,1]^3$ according to
\begin{equation}
    P_{n|s_0}\mapsto \left(x,y,\left[P_{n|s_0}\right]_{i,j}\right), \quad \text{for all }1\leq i,j \leq 2^n
    \label{def2:rho}
\end{equation}
where $(i-1)\cdot 2^{-n}< x < i\cdot 2^{-n}$ and $ 1 - j\cdot 2^{-n}< y < 1 - (j-1)\cdot2^{-n}$. 
\label{def:map}
\end{defin}
Each entry $\left[P_{n|s_0}\right]_{i,j}$ of $P_{n|s_0}$ is identified with a square of side length~$2^{-n}$, which has a distance of  $\left[P_{n|s_0}\right]_{i,j}$ to the $xy$-plane. The alignment of the square corresponding to~$\left[P_{n|s_0}\right]_{i,j}$ with respect to the other squares in $\rho^{(n)}(P_{n|s_0})$ is in accordance to the alignment of~$\left[P_{n|s_0}\right]_{i,j}$ with respect to the other entries of $P_{n|s_0}$. Fig.~\ref{fig:P1} depicts the representations $\rho^{(1)}(P_{1|0})$ and $\rho^{(1)}(P_{1|1})$ of
\begin{equation}
  P_{1|0} = \begin{bmatrix} 1 & 0 \\ \frac{1}{2} & \frac{1}{2}\\\end{bmatrix} \quad
  P_{1|1} = \begin{bmatrix} \frac{1}{2} & \frac{1}{2} \\ 0 & 1\\\end{bmatrix}.\nonumber
\end{equation}
\begin{figure}[ht]
\centering
\subfigure[Color map of  $\rho^{(1)}(P_{1|0})$]{
   \includegraphics[scale =0.45] {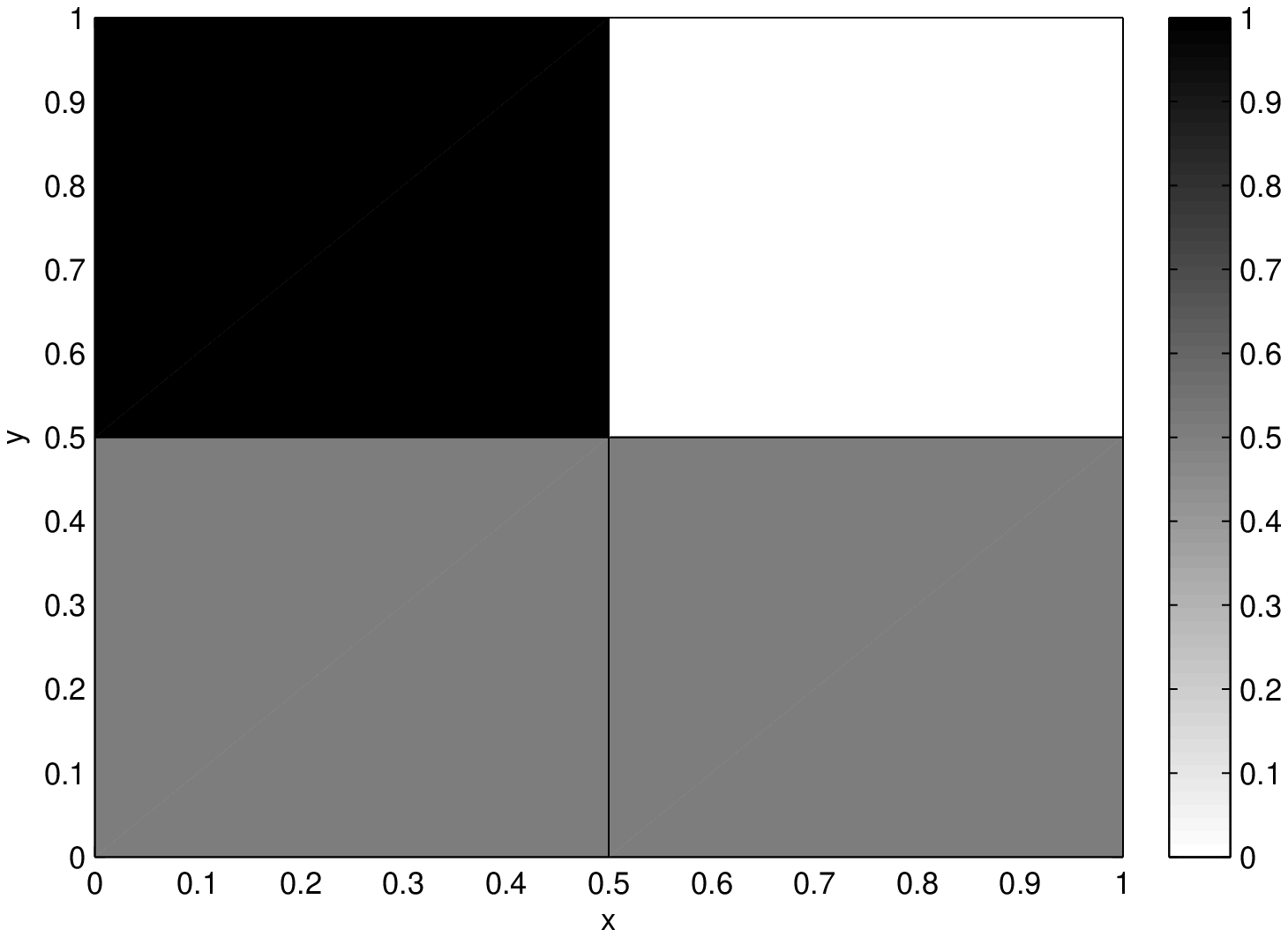}
   \label{fig:P10}
 }

 \subfigure[Color map of $\rho^{(1)}(P_{1|1})$]{
   \includegraphics[scale =0.45] {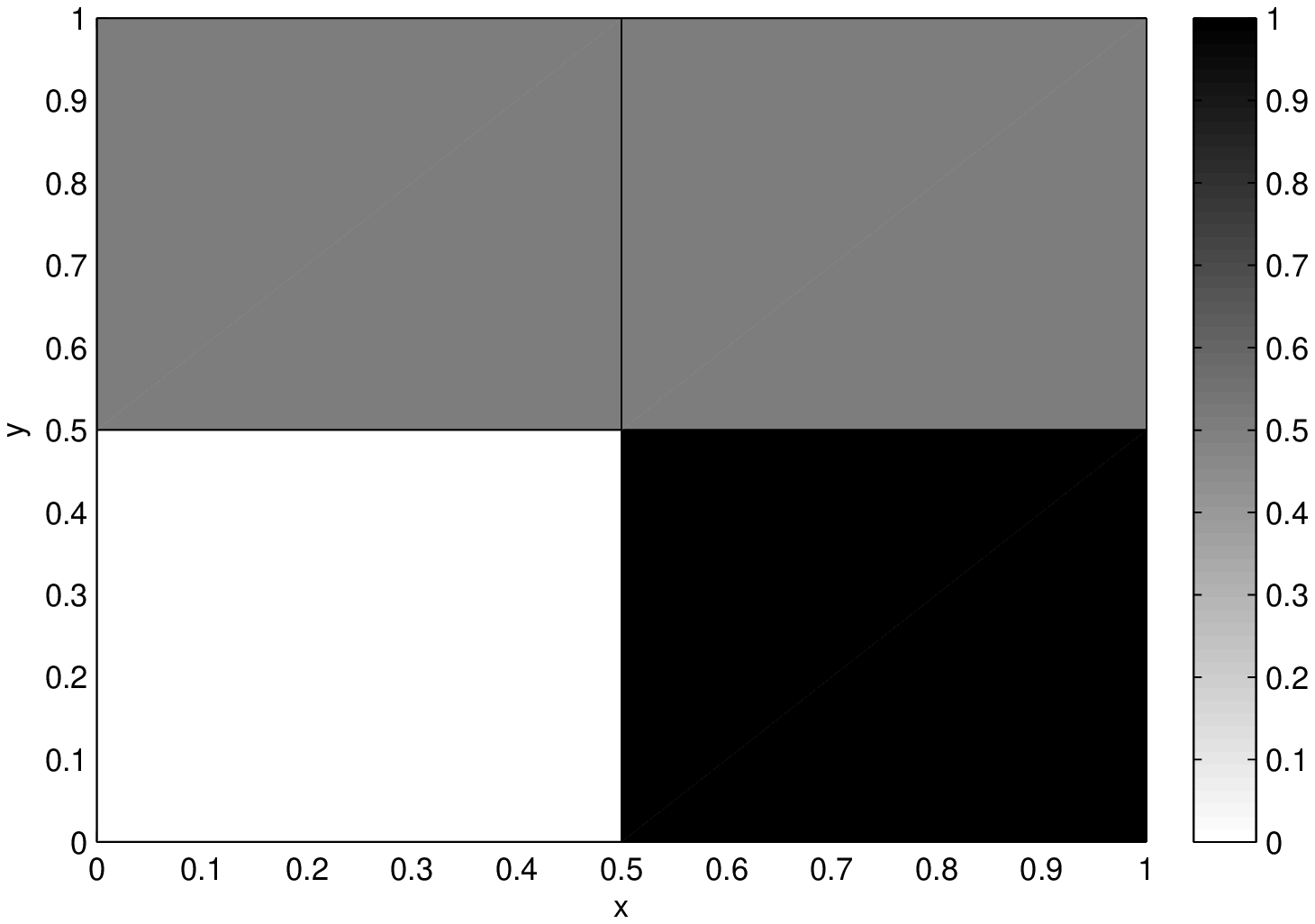}
   \label{fig:P11}
 }
\caption{Color map of the $\rho^{(1)}(P_{1|0})$ and $\rho^{(1)}(P_{1|1})$. Each of the four squares corresponds to one of the conditional probabilities $0, 0.5$ and~$1$.}
\label{fig:P1}
\end{figure}
The following proposition expresses $\rho^{(n+1)}\left(P_{n+1|0}\right)$ and $\rho^{(n+1)}\left(P_{n+1|1}\right)$ recursively in terms of $ \rho^{(n)}\left(P_{n|0}\right)$ and $\rho^{(n)}\left(P_{n|1}\right)$.
\begin{lemma}
  The representations $\rho^{(n+1)}\left(P_{n+1|0}\right)$ and $\rho^{(n+1)}\left(P_{n+1|1}\right)$ of $P_{n+1|0}$ and $P_{n+1|1}$ satisfy the recursion laws
  \begin{align}
     \label{prop1:rec1}
      \rho^{(n+1)}\left(P_{n+1|0}\right) &= \frac{1}{2}\cdot\left\{\rho^{(n)}\left(P_{n|0}\right) + e_x,\hspace{1mm}\rho^{(n)}\left(2\cdot P_{n|0}\right) +  e_y,\hspace{1mm} \rho^{(n)}\left(P_{n|1}\right)\right\}\\
      \rho^{(n+1)}\left(P_{n+1|1}\right) &= \frac{1}{2}\cdot\left\{\rho^{(n)}\left(2\cdot P_{n|1}\right) + e_x,\hspace{1mm} \rho^{(n)}\left(P_{n|1}\right) + e_y,\hspace{1mm} \rho^{(n)}\left(P_{n|0}\right) + e_x + e_y\right\},
      \label{prop1:rec2}
  \end{align}
for all $n \in \mathbb{N}_0$.
\label{prop1}
\end{lemma}
\begin{proof}
Recursions~(\ref{prop1:rec1}) and~(\ref{prop1:rec2}) are a consequence of the structure of block matrices~(\ref{Pn+1|0}) and~(\ref{Pn+1|1}), respectively. We just outline the derivation of~(\ref{prop1:rec1}). The first term on the right hand side of~(\ref{prop1:rec1}) represents the lower right corner of~(\ref{Pn+1|0}), i.e., those entries of $P_{n+1|0}$ with row and column indices $2^n < i,j, \leq 2^{n+1}$. Observe that each entry $\left[P_{n+1|0}\right]_{i,j}$ is equal to $\frac{1}{2}\left[P_{n|0}\right]_{i-2^n,j-2^n}$ where $2^n < i,j, \leq 2^{n+1}$. Hence, scaling the three dimensions of $\rho^{(n)}\left(P_{n|0}\right)$ by a factor of $\frac{1}{2}$ and shifting the result by $\frac{1}{2}$ into the $x$-direction yields a representation of the lower right corner of~(\ref{Pn+1|0}) according to Definition~\ref{def:map}. 

Similarly, the second term of~(\ref{prop1:rec1}) represents the upper left corner of~(\ref{Pn+1|0}), i.e., entries of $P_{n+1|0}$ which correspond to row and column indices $1 \leq i,j, \leq 2^n$. To be more precise, each entry $\left[P_{n+1|0}\right]_{i,j}$ is  equal to $\left[P_{n|0}\right]_{i,j}$ where $1 \leq i,j, \leq 2^n$. Hence, scaling the $x$- and $y$-coordinates of $\rho^{(n)}\left(P_{n|0}\right)$ by a factor of $\frac{1}{2}$ and shifting the resulting figure by $\frac{1}{2}$ into the $y$-direction yields a representation of the upper left corner~$P_{n|0}$ of (\ref{Pn+1|0}) according to Definition~\ref{def:map}.

Finally, the last term of~(\ref{prop1:rec1}) represents the lower left corner of~(\ref{Pn+1|0}), i.e., entries of $P_{n+1|0}$ with row and column indices $2^n < i \leq 2^{n+1}$, $1 \leq j \leq 2^n$, respectively. By~(\ref{Pn+1|0}), each entry $\left[P_{n+1|0}\right]_{i,j}$ is equal to $\frac{1}{2}\left[P_{n|1}\right]_{i-2^n,j}$ for the same index pair~$i,j$. Hence, scaling all coordinates of $\rho^{(n)}\left(P_{n|1}\right)$ by a factor of $\frac{1}{2}$ yields a representation of the lower left corner of~(\ref{Pn+1|0}) according to Definition~\ref{def:map}.
\end{proof}
Recursions~(\ref{prop1:rec1}) and  (\ref{prop1:rec2}) will be used below to obtain an iterated function system for the trapdoor channel. Recall from Theorem~\ref{thm:IFS} that an iterated function system is initialized with a single shape. Therefore, it is desirable that the right hand side of~(\ref{prop1:rec1}) just depends on $P_{n|0}$ and the right hand side of~(\ref{prop1:rec2}) just on $P_{n|1}$. The following proposition introduces an affine transformation, which turns $\rho^{(n)}\left(P_{n|0}\right)$ into $\rho^{(n)}\left(P_{n|1}\right)$ and vice versa.
\begin{lemma}
Let $\tau:[0,1]^3\rightarrow [0,1]^3$ be defined as $\tau(x,y,z) = \left(-x+1,-y+1,z\right)$. Then
\begin{align}
\label{prop2:rec1}
\rho^{(n)}\left(P_{n|1}\right) &= \tau \circ \rho^{(n)}\left(P_{n|0}\right)\\
\rho^{(n)}\left(P_{n|0}\right) &= \tau \circ \rho^{(n)}\left(P_{n|1}\right),
\label{prop2:rec2}
\end{align}
for all $n \in \mathbb{N}_0$.
\label{prop2}
\end{lemma}
\begin{proof}
Equation (\ref{prop2:rec2}) follows from (\ref{prop2:rec1}) by noting that $\tau\circ\tau = id$. It remains to prove (\ref{prop2:rec1}), which we do by induction. Observe that the affine transformation $\tau$ corresponds to a counter-clockwise rotation through 180 degree about the $z$-axis and a translation by one into the $x$- and $y$-direction. Using this property, (\ref{prop2:rec1}) is readily verified from Fig.~\ref{fig:P1} for $n=1$. Now assume that the assertion holds for some $n>1$. A direct computation of~$\tau \circ \rho^{(n+1)}\left(P_{n+1|0}\right)$ using the right hand side of~(\ref{prop1:rec1}) and the induction hypotheses~(\ref{prop2:rec1}) and~(\ref{prop2:rec2}) shows that $\tau \circ \rho^{(n+1)}\left(P_{n+1|0}\right)$ is equivalent to the right hand side of (\ref{prop1:rec2}).
\end{proof}
We can now state the final recursion law. A combination of Lemma~\ref{prop1} and Lemma~\ref{prop2}, i.e., replacing $\rho^{(n)}\left(P_{n|1}\right)$ in (\ref{prop1:rec1}) with (\ref{prop2:rec1}) and $\rho^{(n)}\left(P_{n|0}\right)$ in (\ref{prop1:rec2}) with (\ref{prop2:rec2}), and using~(\ref{def2:rho}) yields the following theorem.
\begin{theorem}
  The representations~$\rho^{(n+1)}\left(P_{n+1|0}\right)$ and~$ \rho^{(n+1)}\left(P_{n+1|1}\right)$ of $P_{n+1|0}$ and $P_{n+1|1}$ with initial matrices $P_{0|0} = P_{0|1} = 1$ satisfy the following recursion laws
  \begin{align}
     \label{thm:rec1}
      \rho^{(n+1)}\left(P_{n+1|0}\right) &= \Bigg\{\phi_1(x,y,z) = \left(\frac{x+1}{2}, \frac{y}{2},\frac{\left[P_{n|0}\right]_{i,j}}{2}\right), \phi_2(x,y,z) = \left(\frac{x}{2},\frac{y+1}{2} ,\left[P_{n|0}\right]_{i,j}\right), \nonumber\\
      &\hspace{0.9cm}\phi_3(x,y,z) = \left(-\frac{x-1}{2},-\frac{y-1}{2},\frac{\left[P_{n|0}\right]_{i,j}}{2}\right)\Bigg\}\\
      \rho^{(n+1)}\left(P_{n+1|1}\right) &= \Bigg\{\psi_1(x,y,z) = \left(\frac{x+1}{2},\frac{y}{2} ,\left[P_{n|1}\right]_{i,j}\right), \psi_2(x,y,z) = \left(\frac{x}{2}, \frac{y+1}{2},\frac{\left[P_{n|1}\right]_{i,j}}{2}\right), \nonumber\\
      &\hspace{0.9cm}\psi_3(x,y,z) = \left(-\frac{x}{2}+1,-\frac{y}{2}+1,\frac{\left[P_{n|1}\right]_{i,j}}{2}\right)\Bigg\},
      \label{thm:rec2}
  \end{align}
where $(i-1)\cdot 2^{-n}< x < i\cdot 2^{-n}$ and $ 1 - j\cdot 2^{-n}< y < 1 - (j-1)\cdot2^{-n}$ for $1\leq i,j \leq 2^n$ .  
\label{thm:rec_trapdoor}
\end{theorem}
\begin{remark}
The restrictions of $\phi_1,\phi_2$, $\phi_3$ and $\psi_1,\psi_2$, $\psi_3$ to the $x$- and $y$-dimensions are contraction mappings. They compose two hyperbolic IFS with a unique attractor each. Moreover, (\ref{thm:rec1}) and ~(\ref{thm:rec2}) are initialized with $P_{0|0} = 1$ and $P_{0|1} = 1$, respectively. Hence, $\lim_{n\rightarrow \infty}\rho^{(n)}\left(P_{n|s_0}\right)$, $s_0\in\{0,1\}$, can be approximated arbitrarily close by iterating (\ref{thm:rec1}) and (\ref{thm:rec2}), respectively, (according to Theorem~\ref{thm:IFS}) for any initial shape $A \in \mathcal{H}([0,1]^3)$ such that the restriction of $A$ to the $z$-dimension equals $1$. Both IFS follow directly from~(\ref{thm:rec1}) and (\ref{thm:rec2}) and read
\begin{align}
 \label{rem:IFS}
 \bigg\{[0,1]^3;\phi_1 = \left(\frac{x+1}{2}, \frac{y}{2},\frac{z}{2}\right), \phi_2 = \left(\frac{x}{2},\frac{y+1}{2} ,z\right),\phi_3 = \left(-\frac{x-1}{2},-\frac{y-1}{2},\frac{z}{2}\right)\bigg\}.\\
 \bigg\{[0,1]^3;\psi_1 = \left(\frac{x+1}{2}, \frac{y}{2},z\right), \psi_2 = \left(\frac{x}{2},\frac{y+1}{2} ,\frac{z}{2}\right),\psi_3 = \left(-\frac{x}{2}+1,-\frac{y}{2}+1,\frac{z}{2}\right)\bigg\}.
\end{align}
There is also a relation to the Sierpinski triangle. Observe that $\phi_1$, $\phi_2$ and $\psi_1$, $\psi_2$, respectively, restricted to the $xy$-plane are equal to $\varphi_1$, $\varphi_2$ in~(\ref{ex:IFS}). 
\end{remark}
\begin{figure}[ht]
\centering
\subfigure[The $z$-dimension is visualized by means of gray colors. The gray scale is the one used in Fig.~\ref{fig:P1}]{
   \includegraphics[scale =0.5] {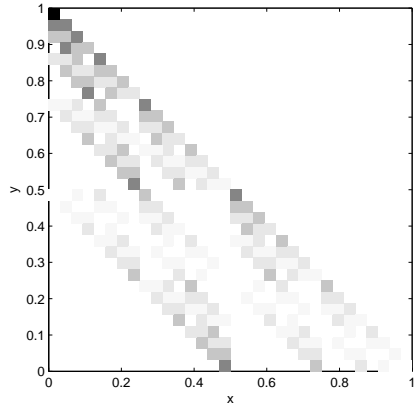}
   \label{fig:trapdoor1}
 }

 \subfigure[Restriction of Fig.~(a) to the $x$- and $y$-dimensions.]{
   \includegraphics[scale =0.5] {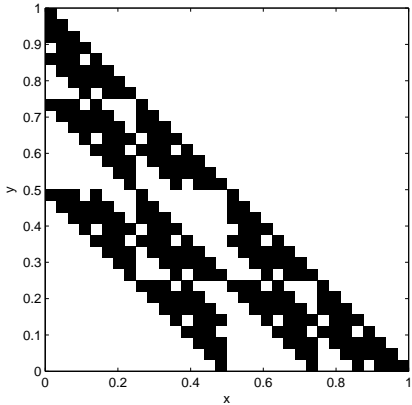}
   \label{fig:trapdoor2}
 }
 \subfigure[A more accurate approximation of the fractal where the IFS~(\ref{rem:IFS}) is restricted to the $x$- and $y$-dimensions.]{
   \includegraphics[scale =0.5] {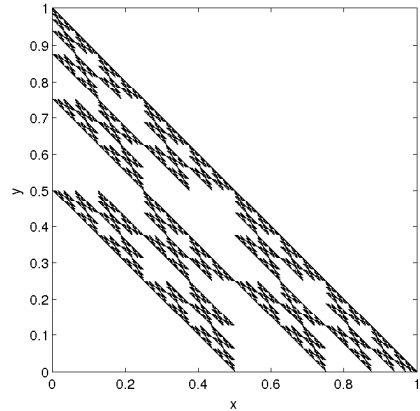}
   \label{fig:limit}
 }
\caption{The result of running $4$ iterations (Fig.~(a), (b)) and $11$ iterations (Fig.~(c)) of the IFS~(\ref{rem:IFS}). The initial shape~$A$ has been chosen to be $\{(x,y,z)\in {[0,1]^3}:z =1\}$.}
\label{fig:trapdoor}
\end{figure}

\section{Algorithmic view of the trapdoor channel}
\label{sec:algorithm}
\subsection{Remarks on the Permutation Nature}
The trapdoor channel has been called a permuting channel~\cite{AhKa87}, where the output is a permutation of the input~\cite{PeHaCuRo08}. We point out that in general not all possible permutations of the input are feasible and that not every output is a permutation of the input. The reason that not all permutations are feasible is that the channel actions are causal, i.e., an input symbol at time~$n$ cannot become a channel output at a time instance smaller than $n$. Consider, for instance, a vector~$101$ which, when applied to a trapdoor channel with initial state~$0$, cannot give rise to an output~$110$. Next, not every output is a permutation of the input because at a certain time instance the initial state might become an output symbol and, therefore, the resulting output sequence might not be compatible with a permutation of the input. For illustration purposes, consider again the previous example, i.e., a vector~$101$ and initial state~$0$. Two of the feasible outputs are~$010$ and~$001$ which are not permutations of~$110$.

\subsection{The Algorithm}
The following recursive procedure {\changefont{ptm}{m}{sc}generateOutputs} computes the set of feasible output sequences and their likelihoods given an input sequence and an initial state. 

\begin{algorithmic}
\Procedure{generateOutputs}{$in, out, state, prob$}
\If{$in = \emptyset$ }
\State $set$ $\gets \{out,prob\}$
\ElsIf{$in[0] = state$}
\State $out \gets out + in[0]$
\State $set \gets$ {\changefont{ptm}{m}{sc}generateOutputs}($in.substr(1), out, state, prob$) 
\Else
\State $out \gets out + in[0]$
\State $set \gets$ {\changefont{ptm}{m}{sc}generateOutputs}($in.substr(1), out, state, 0.5\cdot prob$)
\State $out[out.length()-1] \gets state$ \Comment{$in[0]$ is removed from the end of $out$}
\State $set \gets$ {\changefont{ptm}{m}{sc}generateOutputs}($in.substr(1), out, in[0], 0.5\cdot prob$)
\EndIf
\State \textbf{return} $set$
\EndProcedure
\end{algorithmic}

The four variables $in, out, state$ and $prob$ have the following meaning: $in$ denotes the part of the input string that has not been processed yet; $out$ indicates the part of one particular output string that has been generated so far; $state$ refers to the current channel state; $prob$ denotes the likelihood of $out$. The procedure is initialized with the complete input string and the initial state of the channel; $out$ is initially empty while $prob$ equals~$1$. The first \textbf{if} statement checks the simple case of the recursion, i.e., whether the input string has been processed completely. If yes, then the corresponding output~$out$ and its likelihood~$prob$ is stored and returned in~$set$. Otherwise, we distinguish whether the next input symbol $in[0]$ is equal to the current state. If yes, then the next output takes the value of $in[0]$ (or of $state$ but both are equal), i.e., $out \gets out + in[0]$, with probability $1$ and the procedure {\changefont{ptm}{m}{sc}generateOutputs} is applied recursively to the unprocessed part of the input string, i.e., to $in.substr(1)$, the substring of $in$ with indices greater than~$0$. Clearly, $state$ and $prob$ do not change and, therefore, are passed unmodified to the recursive call. In the other case, i.e., when~$in[0]$ is not equal to the current state, the next output symbol will have a probability of $0.5$ to be either $in[0]$ or $state$. If $in[0]$ becomes the channel output, the following state remains the same. Then the remaining input string $in.substr(1)$ is processed by the recursive call {\changefont{ptm}{m}{sc}generateOutputs}($in.substr(1), out, state, 0.5\cdot prob$). However, if $state$ becomes the channel output, then the following state will be~$in[0]$ and the remaining input string is processed by {\changefont{ptm}{m}{sc}generateOutputs}($in.substr(1), out, in[0], 0.5\cdot prob$). Note that a recursive implementation of the algorithm is needed since it works for inputs of any length, which is not the case if only iterative control structures are used. 

The outlined procedure gives a complete characterization of the trapdoor channel. Generating outputs and their corresponding likelihoods for a particular input sequence might be instrumental for designing codes. Finally, the design of the algorithm resembles a recursion for generating all permutations of a string~(see, e.g., \cite[ch. 8.3]{Rob14}). This gives an algorithmic justification for why some output sequences are permutations of the underlying input sequence.

\section{A Lagrange Multiplier Approach to the Trapdoor Channel}
\label{sec:Lagrange}

\subsection{Problem Formulation}
In this section, we derive an upper bound on the capacity of the trapdoor channel. Specifically, for any $n\in\mathbb{N}$, we find a solution to the optimization problem
\begin{align}
\label{objective}
\text{maximize}&\hspace{0.5cm} \frac{1}{n}I\left(X^n;Y^n|s_0\right)\nonumber\\
&=\frac{1}{n}\sum_{i = 1}^{2^n}\sum_{j = 1}^{2^n}p_i\left[P_{n|s_0}\right]_{i,j}\log\frac{\left[P_{n|s_0}\right]_{i,j}}{\sum_{k=1}^{2^n}p_k\left[P_{n|s_0}\right]_{k,j}}\\
\label{optconstraint1}
\text{subject to}&\hspace{0.5cm} \sum_{i=1}^{2^n} p_i = 1\\
&\hspace{0.5cm}\sum_{k=1}^{2^n}p_k\left[P_{n|s_0}\right]_{k,j}\geq 0\quad \text{for all }1\leq j \leq 2^n.
\label{optconstraint2}
\end{align}
We do not have to distinguish between \textit{lower capacity} and \textit{upper capacity}~\cite[Chapter 4.6]{Gal68} since it does not matter whether the optimization is with respect to inital state~$0$ or~$1$ due to symmetry reasons. Constraint~(\ref{optconstraint2}) guarantees that the argument of the logarithm does not become negative. The feasible set, defined by~(\ref{optconstraint1}) and~(\ref{optconstraint2}), is convex. It includes the set of probability mass functions, but might be larger. To see this note that (\ref{optconstraint2}) is a weighted sum of all $p_k$ where each weight~$\left[P_{n|s_0}\right]_{k,j}$ is nonnegative. Clearly, (\ref{optconstraint1}) and (\ref{optconstraint2}) are satisfied by probability distributions. However, there might exist ``distributions'' which involve negative values and sum up to one but still satisfy (\ref{optconstraint2}). Moreover, the objective function $n^{-1}I\left(X^n;Y^n|s_0\right)$ is concave on the set of probability distributions, which follows by using the same arguments that show that mutual information is concave on the set of input probability distributions. Consequently, the optimization problem is convex and every solution maximizes $n^{-1}I\left(X^n;Y^n|s_0\right)$. In the following, the maximum value is denoted as $C_n^\uparrow$. Taking the limit of the sequence $\left(C_n^\uparrow\right)_{n\in\mathbb{N}}$ as $n$ grows, one obtains either the capacity of the trapdoor channel or an upper bound on the capacity, depending on whether the limit is attained inside or outside the set of probability distributions, respectively. 

\subsection{Using a Result from the Literature}
The reason for considering~(\ref{optconstraint2}) and not the more natural constraints $p_k \geq 0$ for all~$k$ is that a closed form solution can be obtained by applying the method of \textit{Lagrange multipliers} to~(\ref{objective}) and~(\ref{optconstraint1}). In particular, setting the partial derivatives of 
\begin{equation}
\frac{1}{n}I\left(X^n;Y^n|s_0\right) + \lambda \sum_{i=1}^{2^n} p_i
\label{Lagrangian}
\end{equation}
with respect to each of the~$p_i$ equal to zero results in a closed form solution of the considered optimization problem. 

This was done in \cite[Theorem 3.3.3]{Ash65} for general discrete memoryless channels which are square and non singular. Note that $P_{n|s_0}$ is square and non singular (see Lemma~\ref{lem:P2n+2_rec}~(b)). Moreover, we assume that the channel~$P_{n|s_0}$ is memoryless by repeatedly using it over a large number of input blocks of length~$n$. This has the consequence that $C_n^\uparrow$ might be an upper bound on the capacity of a trapdoor channel that is constrained to input blocks of length~$n$. The reason is that some input blocks might drive the channel~$P_{n|s_0}$ into the opposite state $s_0 \oplus 1$, i.e., the upcoming input block would see the channel~$P_{n|s_0 \oplus 1}$ (whose $C_n^\uparrow$ is equal to $C_n^\uparrow$ of $P_{n|s_0}$ by symmetry). However, by assuming that the channel does not change over time, the sender always knows the channel state before a new block is transmitted. Hence, $C_n^\uparrow$ might be an upper bound (even though it is attained on the set of probability distributions). Nevertheless, this issue can be ignored if $n$ goes to infinity because in the asymptotic regime the channel $P_{n|s_0}$ is used only once. But we are interested in the asymptotic regime since the limit of the sequence~$\left(C_n^\uparrow\right)_{n\in\mathbb{N}}$ is also its supremum (see Theorem~\ref{thm:main}). 

In summary, we can apply \cite[Theorem 3.3.3]{Ash65} which yields
\begin{equation}
C_n^{\uparrow} = \frac{1}{n}\log_2\sum_{j=1}^{2^n}\exp_2\left(-\sum_{i=1}^{2^n}\begin{bmatrix}P_{n|s_0}^{-1}\end{bmatrix}_{j,i} H(Y^n|X^n = x_i^n)\right),
\label{thm:ash:eq2}
\end{equation}
attained at
\begin{equation}
p_i = 2^{-C_n^\uparrow}d_i,\quad i = 1,2,\dots,2^n
\end{equation}
where $d_i$ equals
\begin{equation}
\sum_{j=1}^{2^n} \begin{bmatrix}P_{n|s_0}^{-1}\end{bmatrix}_{j,k}\exp_2\left(-\sum_{i=1}^M \begin{bmatrix}P_{n|s_0}^{-1}\end{bmatrix}_{j,i}H(Y^n|X^n = x_i^n)\right). 
\label{thm:ash:condition}
\end{equation}
Clearly, $\begin{bmatrix}p_1,\dots,p_{2^n}\end{bmatrix}$ is a probability distribution only if $d_i\geq 0$. Observe that the Lagrangian~(\ref{Lagrangian}) does not involve the constraint~(\ref{optconstraint2}). However, the proof of\cite[Theorem 3.3.3]{Ash65} shows that $\sum_{k=1}^{2^n}p_k\left[P_{n|s_0}\right]_{k,j}$ equals
\begin{equation}
\exp\left(\lambda -\sum_{i=1}^M \begin{bmatrix}P_{n|s_0}^{-1}\end{bmatrix}_{j,i}H(Y^n|X^n = \mathbf{x}_i) -1\right)  
\end{equation}
for all $1\leq j\leq 2^n$. Hence, (\ref{optconstraint2}) is satisfied.

We remark that (\ref{thm:ash:eq2}) in matrix notation reads
\begin{equation}
C_n^\uparrow = \frac{1}{n}\log_2\left[1_n^T\exp_2\left(P_{n|s_0}^{-1}\left(P_{n|s_0}\circ\log_2 P_{n|s_0}\right)1_n\right)\right].
\label{C:ash:matrix}
\end{equation}
In the remainder, we will evaluate~(\ref{C:ash:matrix}).

\subsection{Useful Recursions}
To evaluate~(\ref{C:ash:matrix}), we derive recursions for $-\left(P_{n|s_0}\circ\log_2 P_{n|s_0}\right)1_{n}$ and $P_{n|s_0}^{-1}\left(P_{n|s_0}\circ\log_2 P_{n|s_0}\right)1_{n}$. The two expressions are formally defined next. Based on these recursions, we find exact numerical expressions for~(\ref{C:ash:matrix}) in Theorem~\ref{thm:main} below. 
\begin{defin}
(a) The \textnormal{conditional entropy vector} $h_{n|s_0}$ of $P_{n|s_0}$, $s_0\in\{0,1\}$, is defined as 
\begin{align}
h_{n|s_0}&=\begin{bmatrix} H(Y^n|X^n = x^n_1) & \dots & H(Y^n|X^n = x^n_{2^{n}})\end{bmatrix}^T\\
\label{hadamard}
&=-\left(P_{n|s_0}\circ\log_2 P_{n|s_0}\right)1_{n}
\end{align}
where $n \in \mathbb{N}_0$.\\
(b) The \textnormal{weighted conditional entropy vector} $\omega_{n|s_0}$ of $P_{n|s_0}$, $s_0\in\{0,1\}$, is defined as
\begin{align}
  \label{proof:omega_first}
 \omega_{n|s_0}&=-P_{n|s_0}^{-1}\cdot h_{n|s_0}\\
 &=P_{n|s_0}^{-1}\left(P_{n|s_0}\circ\log_2 P_{n|s_0}\right)1_{n}
 \label{proof:omega}
\end{align}
where $n \in \mathbb{N}_0$.
\label{def:cond}
\end{defin}
We remark that $h_{n|s_0}$ and $\omega_{n|s_0}$ are column vectors with $2^n$ entries. The following two lemmas provide tools that we need for the proof of Lemma~\ref{lemma:cond_entropy_rec} and Lemma~\ref{lemma:weighted_cond_entropy_rec}. 
\begin{lemma}
(a) The trapdoor channel matrices $P_{2n+2|0}$ and $P_{2n+2|1}$, $n \in \mathbb{N}_0$, satisfy the following recursions:
 \begin{align}
    \label{P_{2n+2|0}}
    P_{2n+2|0} &= \begin{bmatrix} P_{2n|0} & 0 & 0 & 0 \\ \frac{1}{2}P_{2n|1} & \frac{1}{2}P_{2n|0} & 0 & 0\\ \frac{1}{4}P_{2n|1} & \frac{1}{4}P_{2n|0} & \frac{1}{2}P_{2n|0} & 0 \\ 0 & \frac{1}{2}P_{2n|1} & \frac{1}{4}P_{2n|1} & \frac{1}{4}P_{2n|0}\end{bmatrix}\\
    \label{P_{2n+2|1}}
    P_{2n+2|1} &= \begin{bmatrix}\frac{1}{4}P_{2n|1} & \frac{1}{4}P_{2n|0} & \frac{1}{2}P_{2n|0} & 0\\ 0 & \frac{1}{2}P_{2n|1} & \frac{1}{4}P_{2n|1} & \frac{1}{4}P_{2n|0} \\ 0 & 0 & \frac{1}{2}P_{2n|1} & \frac{1}{2}P_{2n|0} \\ 0 & 0 & 0 & P_{2n|1} \end{bmatrix}.
  \end{align}
(b) Let $M_0 := P^{-1}_{2n|0}P_{2n|1}P^{-1}_{2n|0}$ and $M_1 := P^{-1}_{2n|1}P_{2n|0}P^{-1}_{2n|1}$. The inverses of $P_{2n+2|0}$ and $P_{2n+2|1}$, $n \in \mathbb{N}_0$, satisfy the following recursions:
 \begin{align}
    \label{P^{-1}_{2n+2|0}}
    P^{-1}_{2n+2|0} &= \begin{bmatrix} P^{-1}_{2n|0} & 0 & 0 & 0 \\ -M_0 & 2 P^{-1}_{2n|0} & 0 & 0\\ 0 & -P^{-1}_{2n|0} & 2P^{-1}_{2n|0} & 0 \\ 2 M_0 P_{2n|1}P^{-1}_{2n|0} & -3 M_0 & -2M_0 & 4P^{-1}_{2n|0}\end{bmatrix}
  \end{align}
  \begin{align}
    \label{P^{-1}_{2n+2|1}}
    P^{-1}_{2n+2|1} &= \begin{bmatrix} 4 P^{-1}_{2n|1} & -2 M_1 & -3 M_1 & 2 M_1P_{2n|0}P^{-1}_{2n|1}  \\ 0 & 2 P^{-1}_{2n|1} & -P^{-1}_{2n|1} & 0\\ 0 & 0 & 2P^{-1}_{2n|1} & -M_1 \\ 0 & 0 & 0 & P^{-1}_{2n|1}\end{bmatrix}.
  \end{align}
  \label{lem:P2n+2_rec}
\end{lemma}
\begin{proof}
(a): Substituting $P_{2n+2-1|0}$ and $P_{2n+2-1|1}$ into $P_{2n+2|0}$ and $P_{2n+2|1}$, where the four matrices are expressed as in~(\ref{Pn+1|0}) and~(\ref{Pn+1|1}), yields~(\ref{P_{2n+2|0}}) and~(\ref{P_{2n+2|1}}).\\
(b): Two versions of the matrix inversion lemma are~\cite{GolLoa96}
\begin{align}
\label{invLemma1}
\begin{bmatrix} A & 0 \\ C & D\end{bmatrix}^{-1} &= \begin{bmatrix}A^{-1} & 0 \\ -D^{-1}CA^{-1} & D^{-1} \end{bmatrix}\\
\label{invLemma2}
\begin{bmatrix} A & B \\ 0 & D\end{bmatrix}^{-1} &= \begin{bmatrix}A^{-1} & -A^{-1}BD^{-1} \\ 0 & D^{-1} \end{bmatrix}.                                                 
\end{align}
Divide (\ref{P_{2n+2|0}}) and~(\ref{P_{2n+2|1}}) into four blocks of equal size. A twofold application of~(\ref{invLemma1}) and~(\ref{invLemma2}), first to $P_{2n+2|0}$ and $ P_{2n+2|1}$ and, subsequently, to each of the blocks of $P_{2n+2|0}$ and $ P_{2n+2|1}$ yields~(\ref{P^{-1}_{2n+2|0}}) and~(\ref{P^{-1}_{2n+2|1}}).
\end{proof}
A transformation relating~$P_{n|0}$ with $P_{n|1}$, $P_{n|0}^{-1}$ with $P_{n|1}^{-1}$, $h_{n|0}$ with $h_{n|1}$ and $\omega_{n|0}$ with $\omega_{n|1}$ is derived next.
\begin{lemma}
Let $P_{n|0}$ and $P_{n|1}$ be trapdoor channel matrices, $n \in \mathbb{N}_0$. Then we have the following identities.\\
(a)\begin{align}
\label{trafo1}
 P_{n|1} &= \tilde{I}_n P_{n|0}\tilde{I}_n\\
 \label{trafo2}
 P_{n|0} &= \tilde{I}_n P_{n|1}\tilde{I}_n.
\end{align}
(b)\begin{align}
\label{invtrafo1}
 P_{n|1}^{-1} &= \tilde{I}_n P_{n|0}^{-1}\tilde{I}_n\\
 \label{invtrafo2}
 P_{n|0}^{-1} &= \tilde{I}_n P_{n|1}^{-1}\tilde{I}_n.
 \end{align}
(c)\begin{align}
\label{trafo1_h}
 h_{n|1} &= \tilde{I}_n h_{n|0}\\
 \label{trafo2_h}
 h_{n|0} &= \tilde{I}_n h_{n|1}.
\end{align}
(d)\begin{align}
\label{trafo1_w}
 \omega_{n|1} &= \tilde{I}_n \omega_{n|0}\\
 \label{trafo2_w}
 \omega_{n|0} &= \tilde{I}_n \omega_{n|1}.
\end{align}
(e) The row sums of $P^{-1}_{n|0}$ and $P^{-1}_{n|1}$ are~$1$.
\label{lem:trafo}
\end{lemma}
\begin{proof}
(a): The proof is by induction. For $n=0$, the identities $P_{0|1} = \tilde{I}_0 P_{0|0}\tilde{I}_0$ and $P_{0|0} = \tilde{I}_0 P_{0|1}\tilde{I}_0$ clearly hold. Now suppose that~(\ref{trafo1}) and~(\ref{trafo2}) are true if $n$ is replaced by $n-1$. Then we have
\begin{align}
 \label{proof-trafo:1}
 \tilde{I}_n P_{n|0}\tilde{I}_n &= \begin{bmatrix}0 & \tilde{I}_{n-1}\\ \tilde{I}_{n-1} & 0\end{bmatrix}\begin{bmatrix}P_{n-1|0} & 0\\\frac{1}{2}P_{n-1|1} & \frac{1}{2}P_{n-1|0}\end{bmatrix}\begin{bmatrix}0 & \tilde{I}_{n-1}\\ \tilde{I}_{n-1} & 0\end{bmatrix} \\
 &= \begin{bmatrix}\frac{1}{2}\tilde{I}_{n-1} P_{n-1|0}\tilde{I}_{n-1} & \frac{1}{2}\tilde{I}_{n-1} P_{n-1|1}\tilde{I}_{n-1} \\0 & \tilde{I}_{n-1} P_{n-1|0}\tilde{I}_{n-1}\end{bmatrix}\nonumber
  \end{align}
 \begin{align}
 \label{proof-trafo:2}
 &= \begin{bmatrix}\frac{1}{2}P_{n-1|1} & \frac{1}{2}P_{n-1|0}\\0 & P_{n-1|1}\end{bmatrix}\\
 \label{proof-trafo:3}
 &= P_{n-1|1}
\end{align}
where (\ref{proof-trafo:1}) and~(\ref{proof-trafo:3}) are due to the recursive expressions~(\ref{Pn+1|0}) and~(\ref{Pn+1|1}) while (\ref{proof-trafo:2}) follows from the induction hypothesis. It remains to show~(\ref{trafo2}). But~(\ref{trafo2}) is a direct consequence of the just proven equation and using the identity $\tilde{I}_n\tilde{I}_n = I_n$.\\
(b): Follows immediately from (a) and the identity $\tilde{I}_n\tilde{I}_n = I_n$.\\ 
(c): Equation (\ref{trafo1_h}) follows from 
\begin{align}
  h_{n|1}&=-\left(P_{n|1}\circ\log_2 P_{n|1}\right)1_{n}\nonumber\\
  \label{pr:trafo2_h}
  &= -\left[\left(\tilde{I}_nP_{n|0}\tilde{I}_n\right)\circ\log_2 \left(\tilde{I}_nP_{n|0}\tilde{I}_n\right)\right]1_{n}\\
  \label{pr:trafo3_h}
  &= -\tilde{I}_n\left(P_{n|0}\circ\log_2 P_{n|0}\right)\tilde{I}_n1_{n}\\
  &= \tilde{I}_n h_{n|0}\nonumber
\end{align}
where (\ref{pr:trafo2_h}) follows by replacing~$P_{n|1}$ with~(\ref{trafo1}). Observe that the left and right multiplication of $P_{n|0}$ with~$\tilde{I}_n$ merely yields a new ordering of the elements of $P_{n|0}$.\footnote{To be more precise, $[P_{n|0}]_{i,j}$ is placed at position~$(2^n +1-i, 2^n +1-j)$ for all~$1 \leq i,j \leq 2^n$.} Since it does not matter whether the Hadamard product and the elementwise logarithm is applied before or after sorting the elements of the underlying matrix, i.e., before or after multiplying with $\tilde{I}_{2n}$,~(\ref{pr:trafo3_h}) is true.

Equation~(\ref{trafo2_h}) follows from (\ref{trafo1_h}) and the identity $\tilde{I}_n\tilde{I}_n = I_n$.\\
(d): Equation (\ref{trafo1_w}) follows from
\begin{align}
  \omega_{n|1}&=-P_{n|1}^{-1}h_{n|1}\nonumber\\
  \label{pr:trafo_w}
  &= -\tilde{I}_nP_{n|0}^{-1}h_{n|0}\\
  &= \tilde{I}_n \omega_{n|0},\nonumber
\end{align}
where (\ref{pr:trafo_w}) follows by replacing~$P_{n|1}$ and~$h_{n|1}$ with~(\ref{trafo1}) and~(\ref{trafo1_h}), respectively, and using the identity~$\tilde{I}_n\tilde{I}_n = I_n$.

Equation~(\ref{trafo2_w}) follows from (\ref{trafo1_w}) and the identity $\tilde{I}_n\tilde{I}_n = I_n$.\\
(e): A standard way to compute $P_{n|0}^{-1}$ is by Gauss-Jordan elimination, i.e., a sequence of elementary row operations applied to the augmented matrix $\begin{bmatrix}P_{n|0} & I_n\end{bmatrix}$ such that $\begin{bmatrix}I_n & P_{n|0}^{-1}\end{bmatrix}$ eventually results. Clearly, $P_{n|0}$ and $I_n$ are stochastic matrices, i.e., all row sums are equal to one. Thus, at each stage of performing the elementary row operations, the row sum of the left matrix equals the row sum of the right matrix. In particular, $P_{n|0}^{-1}$ has the same row sum as $I_n$.
\end{proof}
We can now state the recursive laws for the \textit{conditional entropy vector} and the \textit{weighted conditional entropy vector}.
\begin{lemma}
For $n\geq1$, $h_{2n+2|0}$ satisfies the recursion
\begin{equation}
h_{2n+2|0} = \begin{bmatrix}h_{2n|0} \\ \frac{1}{2}h_{2n|0} + \frac{1}{2}\tilde{I}_{2n}h_{2n|0} + 1_{2n} \\ \frac{3}{4}h_{2n|0} + \frac{1}{4}\tilde{I}_{2n}h_{2n|0} + \frac{3}{2} 1_{2n}\\ \frac{1}{4}h_{2n|0} + \frac{3}{4}\tilde{I}_{2n}h_{2n|0} + \frac{3}{2}1_{2n}\end{bmatrix}.
\label{cond_entropy_rec}
\end{equation}
The initial value for~$n=0$ is given by $h_{0|0} = 0$. 
\label{lemma:cond_entropy_rec}
\end{lemma}
We remark that in order to refer to the $i$th subvector, $1\leq i \leq 4$, of the conditional entropy vector we use the superscript~$(i)$. For instance, $h_{2n+2|0}^{(2)}$ refers to $\frac{1}{2}h_{2n|0} + \frac{1}{2}\tilde{I}_{2n}h_{2n|0} + 1_{2n}$.
\begin{proof}
The initial value $h_{0|0}$ can be directly computed using~$P_{0|0}=1$ in~(\ref{hadamard}). In order to show~(\ref{cond_entropy_rec}), we replace $P_{2n+2|0}$ in~(\ref{hadamard}) with~(\ref{P_{2n+2|0}}) from Lemma~\ref{lem:P2n+2_rec}~(a) and compute each of the four entries in~(\ref{cond_entropy_rec}) separately. Clearly, we have $h_{2n+2|0}^{(1)} = -\left(P_{2n|0}\circ\log_2 P_{2n|0}\right)1_{2n}$, which by definition equals $h_{2n|0}$. The three remaining terms can be written as follows
\begin{align}
h_{2n+2|0}^{(2)}&=\left[-\frac{1}{2}P_{2n|1}\circ\log_2 \left(\frac{1}{2}P_{2n|1}\right) - \frac{1}{2}P_{2n|0}\circ\log_2 \left(\frac{1}{2}P_{2n|0}\right)\right]1_{2n}\nonumber\\
\label{h2:1}
&=\left[\frac{1}{2}P_{2n|1} - \frac{1}{2}\left(\tilde{I}_{2n}P_{2n|0}\tilde{I}_{2n}\right)\circ\log_2 \left(\tilde{I}_{2n}P_{2n|0}\tilde{I}_{2n}\right) + \frac{1}{2}P_{2n|0} - \frac{1}{2}P_{2n|0}\circ\log_2 P_{2n|0} \right]1_{2n}\\
\label{h2:2}
&=1_{2n} - \frac{1}{2}\tilde{I}_{2n}\left(P_{2n|0}\circ\log_2 P_{2n|0}\right)1_{2n} +  \frac{1}{2}h_{2n|0}\\
&= \frac{1}{2}h_{2n|0} + \frac{1}{2}\tilde{I}_{2n}h_{2n|0} + 1_{2n}\nonumber
\end{align}
\begin{align}
h_{2n+2|0}^{(3)}&=\left[-\frac{1}{4}P_{2n|1}\circ\log_2 \left(\frac{1}{4}P_{2n|1}\right) - \frac{1}{4}P_{2n|0}\circ\log_2 \left(\frac{1}{4}P_{2n|0}\right)- \frac{1}{2}P_{2n|0}\circ\log_2 \left(\frac{1}{2}P_{2n|0}\right)\right]1_{2n}\nonumber\\
\label{h3:1}
&=\left[\frac{1}{2}P_{2n|1} - \frac{1}{4}\left(\tilde{I}_{2n}P_{2n|0}\tilde{I}_{2n}\right)\circ\log_2 \left(\tilde{I}_{2n}P_{2n|0}\tilde{I}_{2n}\right) + P_{2n|0} - \frac{3}{4}P_{2n|0}\circ\log_2 P_{2n|0} \right]1_{2n}\\
\label{h3:2}
&=\frac{3}{2} 1_{2n} - \frac{1}{4}\tilde{I}_{2n}\left(P_{2n|0}\circ\log_2 P_{2n|0}\right)1_{2n} +  \frac{3}{4}h_{2n|0}\\
&= \frac{3}{4}h_{2n|0} + \frac{1}{4}\tilde{I}_{2n}h_{2n|0} + \frac{3}{2} 1_{2n}\nonumber
\end{align}
\begin{align}
h_{2n+2|0}^{(4)}&=\left[-\frac{1}{2}P_{2n|1}\circ\log_2 \left(\frac{1}{2}P_{2n|1}\right) - \frac{1}{4}P_{2n|1}\circ\log_2 \left(\frac{1}{4}P_{2n|1}\right) - \frac{1}{4}P_{2n|0}\circ\log_2 \left(\frac{1}{4}P_{2n|0}\right)\right]1_{2n}\nonumber\\
\label{h4:1}
&=\left[P_{2n|1} - \frac{3}{4}\left(\tilde{I}_{2n}P_{2n|0}\tilde{I}_{2n}\right)\circ\log_2 \left(\tilde{I}_{2n}P_{2n|0}\tilde{I}_{2n}\right) + \frac{1}{2}P_{2n|0} - \frac{1}{4}P_{2n|0}\circ\log_2 P_{2n|0}\right]1_{2n}
\end{align}
\begin{align}
\label{h4:2}
&=\frac{3}{2} 1_{2n} - \frac{3}{4}\tilde{I}_{2n}\left(P_{2n|0}\circ\log_2 P_{2n|0}\right)1_{2n} +  \frac{1}{4}h_{2n|0} - \frac{3}{4}\tilde{I}_{2n}\left(P_{2n|0}\circ\log_2 P_{2n|0}\right)1_{2n}\\
&= \frac{1}{4}h_{2n|0} + \frac{3}{4}\tilde{I}_{2n}h_{2n|0} + 1_{2n}\nonumber
\end{align}
where~(\ref{h2:1}), (\ref{h3:1}) and (\ref{h4:1}), respectively, follow from expanding the logarithms in the previous equation and replacing the channel matrices corresponding to initial state one with~(\ref{trafo1}). The first term in~(\ref{h2:2}), (\ref{h3:2}) and (\ref{h4:2}), respectively, follows from the multiplication of the weighted matrices~$P_{2n|0}$ and~$P_{2n|1}$ with $1_n$. The second term in~(\ref{h2:2}), (\ref{h3:2}) and (\ref{h4:2}), respectively, follows by using the fact that it does not matter whether the Hadamard product and the elementwise logarithm is applied before or after sorting the elements of the underlying matrix, i.e., before or after multiplying with $\tilde{I}_{2n}$.
\end{proof}
\begin{lemma}
(a) For $n \geq 1$, $\omega_{2n|0}$ satisfies the recursion
\begin{equation}
 \omega_{2n|0} = \begin{bmatrix}\omega_{2n-2|0} \\ \omega_{2n-2|0} - 2\cdot 1_{2n-2} \\ \omega_{2n-2|0} - 2\cdot 1_{2n-2} \\ \omega_{2n-2|0} \end{bmatrix}
 \label{proof:omega_rec}
\end{equation}
with initial value $\omega_{0|0} = 0$.\\
(b) For $n \geq 1$, $\omega_{2n+1|0}$ satisfies the recursion
\begin{equation}
 \omega_{2n+1|0} = \begin{bmatrix}\omega_{2n-1|0} \\ \tilde{I}_{2n-1}\omega_{2n-1|0} \\ \omega_{2n-1|0} - 2 \cdot1_{2n-1} \\ \tilde{I}_{2n-1}\omega_{2n-1|0} - 2\cdot 1_{2n-1} \end{bmatrix}
 \label{proof:omega:rec_odd}
\end{equation}
with initial value $\omega_{1|0} = \begin{bmatrix}0 & -2\end{bmatrix}^T$.
\label{lemma:weighted_cond_entropy_rec}
\end{lemma}
We remark that in order to refer to the $i$th subvector, $1\leq i \leq 4$, of the weighted conditional entropy vector we use the superscript~$(i)$. For instance, $\omega_{2n|0}^{(2)}$ refers to $\omega_{2n-2|0} - 2\cdot 1_{2n-2}$.
\begin{proof}
(a): We first show by induction that (\ref{proof:omega_rec}) holds. The case $n = 0$ can be verified using Definition~\ref{def:cond}~(b) with $P_{0|0} = P_{0|0}^{-1} = 1$. Now assume that~(\ref{proof:omega_rec}) holds for some~$n$.  In order to show~(\ref{proof:omega_rec}) for $n+1$, we evaluate $\omega_{2n+2|0}$ using~(\ref{proof:omega}) and replacing~$P_{2n+2|0}^{-1}$ and $h_{2n+2|0}$ with~(\ref{P^{-1}_{2n+2|0}}) and~(\ref{cond_entropy_rec}). Then we have
\begin{equation}
\omega_{2n+2|0} = \begin{bmatrix}-P_{2n|0}^{-1}h_{2n+2|0}^{(1)}\\P_{2n|0}^{-1}\left(P_{2n|1}P_{2n|0}^{-1}h_{2n+2|0}^{(1)} - 2 h_{2n+2|0}^{(2)}\right)\\ P_{2n|0}^{-1} \left(h_{2n+2|0}^{(2)} - 2h_{2n+2|0}^{(3)}\right)\\ M_0\left(-2P_{2n|1}P_{2n|0}^{-1}h_{2n+2|0}^{(1)} + 3h_{2n+2|0}^{(2)} + 2h_{2n+2|0}^{(3)}\right) - 4P_{2n|0}^{-1}h_{2n+2|0}^{(4)}\end{bmatrix}. 
\label{proof:omega_rec_expl}
\end{equation}
Recall from Lemma~\ref{lemma:cond_entropy_rec} that~$h_{2n+2|0}^{(1)} = h_{2n|0}$. Hence, by definition, the first entry of~(\ref{proof:omega_rec_expl}) is equal to $\omega_{2n|0}$. 

The second entry of~(\ref{proof:omega_rec_expl}) is derived as follows. Replacing~$h_{2n+2|0}^{(1)}$ and~$h_{2n+2|0}^{(2)}$ with the corresponding expressions from~(\ref{cond_entropy_rec}), we obtain
\begin{equation}
\omega_{2n+2|0}^{(2)} = P_{2n|0}^{-1}\left(P_{2n|1}P_{2n|0}^{-1}h_{2n|0} - h_{2n|0} - \tilde{I}_{2n} h_{2n|0} - 2\cdot1_{2n}\right).
\label{w2n+2(2)}
\end{equation}
In order to simplify~(\ref{w2n+2(2)}), observe that
\begin{equation}
\label{implication1}
- \tilde{I}_{2n}\omega_{2n|0} + \omega_{2n|0} = 0
\end{equation}
since $\omega_{2n|0}$ is a palindromic vector by hypothesis. A further manipulation of~(\ref{implication1}), namely using (\ref{proof:omega_first}), (\ref{invtrafo2}) and the relation~$\tilde{I}_{2n}\tilde{I}_{2n} = I_{2n}$, yields
\begin{equation}
\label{implication2}
P_{2n|0}^{-1}\cdot h_{2n|0} - P_{2n|1}^{-1}\tilde{I}_{2n}\cdot h_{2n|0}  = 0
\end{equation}
which implies
\begin{equation}
\label{implication3}
P_{2n|1}P_{2n|0}^{-1}h_{2n|0} - \tilde{I}_{2n}\cdot h_{2n|0} = 0.
\end{equation}
Using~(\ref{implication3}), the definition of~$\omega_{2n|0}$ and  Lemma~\ref{lem:trafo}~(e), i.e., that $P_{2n|0}^{-1}$ is a stochastic matrix, in~(\ref{w2n+2(2)}) we obtain~$\omega_{2n+2|0}^{(2)}=\omega_{2n|0} - 2\cdot 1_{2n}$. 

The third entry of~(\ref{proof:omega_rec_expl}) is derived as follows. After replacing~$h_{2n+2|0}^{(2)}$ and~$h_{2n+2|0}^{(3)}$ in~(\ref{proof:omega_rec_expl}) with the corresponding expressions from~(\ref{cond_entropy_rec}), it can be directly seen that~$\omega_{2n+2|0}^{(3)} = \omega_{2n|0} - 2\cdot 1_{2n}$. 

Regarding the fourth entry in~(\ref{proof:omega_rec_expl}), we begin with the first term in parentheses, i.e.,
\begin{align}
&-2P_{2n|1}P_{2n|0}^{-1}h_{2n+2|0}^{(1)} + 3h_{2n+2|0}^{(2)} + 2h_{2n+2|0}^{(3)}\nonumber\\
\label{w2n+2(4)1}
=&-2\left(P_{2n|1}P_{2n|0}^{-1}h_{2n+2|0}^{(1)} - 2h_{2n+2|0}^{(2)}\right) - \left(h_{2n+2|0}^{(2)} - 2h_{2n+2|0}^{(3)}\right) \\
\label{w2n+2(4)2}
=& -3 P_{2n|0}\left(\omega_{2n|0} - 2\cdot 1_{2n}\right).
\end{align}
Equation~(\ref{w2n+2(4)2}) holds since the first and the second parentheses of~(\ref{w2n+2(4)1}) are equal to $P_{2n|0}\omega_{2n+2|0}^{(2)}$ and~$P_{2n|0}\omega_{2n+2|0}^{(3)}$, respectively, which follows from~(\ref{proof:omega_rec_expl}) by inspection. Moreover, $\omega_{2n+2|0}^{(2)}$ and $\omega_{2n+2|0}^{(3)}$ are equal to~$\omega_{2n|0} - 2\cdot 1_{2n}$ as we just have shown. Hence, using~(\ref{w2n+2(4)2}) in~$\omega_{2n+2|0}^{(4)}$ and replacing $h_{2n+2|0}^{(4)}$ with the corresponding expression from~(\ref{cond_entropy_rec}) and $M_0$ with its definition from Lemma~\ref{lem:P2n+2_rec}~(b), we obtain
\begin{align}
\omega_{2n+2|0}^{(4)} &= P_{2n|0}^{-1}\left(-3P_{2n|1}\left(\omega_{2n|0} - 2\cdot 1_{2n}\right) - h_{2n|0} - 3\tilde{I}_{2n}h_{2n|0} - 6\cdot1_{2n}\right)\nonumber\\
\label{step1}
&= 3P_{2n|0}^{-1}\left(-P_{2n|1}\omega_{2n|0} - \tilde{I}_{2n}h_{2n|0}\right) + 6P_{2n|0}^{-1}\left(P_{2n|1}1_{2n} - 1_{2n}\right) -P_{2n|0}^{-1}h_{2n|0}\\
&= -P_{2n|0}^{-1}h_{2n|0} \nonumber\\
&= \omega_{2n|0}.\nonumber
\end{align}
Observe that the first parentheses in~(\ref{step1}), which is equal to the left hand side of~(\ref{implication3}), evaluates to~$0$. Also the second parentheses in~(\ref{step1}) evaluates to~$0$ since~$P_{2n|1}$ is a stochastic matrix.

(b): Recall the recursions
\begin{align}
 \label{step2:1}  
P_{2n+2|0} &= \begin{bmatrix} P_{2n+1|0} & 0 \\ \frac{1}{2}P_{2n+1|1} & \frac{1}{2}P_{2n+1|0}\\\end{bmatrix}\\               
P_{2n+2|0}^{-1} &= \begin{bmatrix} P_{2n+1|0}^{-1} & 0 \\ P_{2n+1|0}^{-1}P_{2n+1|1}P_{2n+1|0}^{-1} & 2P_{2n+1|0}^{-1}\\\end{bmatrix}.
 \label{step2:2}                                                                                                 
\end{align}
The first $2^{2n+1}$ entries, i.e., the first half, of~$\omega_{2n+2|0}$ are equal to $P_{2n+1|0}^{-1}\left(P_{2n+1|0}\circ\log_2 P_{2n+1|0}\right)1_{2n+1}$, which in turn is equal to~$\omega_{2n+1|0}$. This follows from a straightforward computation using Definition~\ref{def:cond}(b) together with~(\ref{step2:1}) and~(\ref{step2:2}). Hence, under consideration of~(\ref{proof:omega_rec}), we have
\begin{equation}
 \omega_{2n+1|0} = \begin{bmatrix}\omega_{2n|0}\\ \omega_{2n|0} - 2\cdot1_{2n}\end{bmatrix}.
  \label{step2:3}   
\end{equation}
Equivalently, $\omega_{2n-1|0}$ is equal to the first $2^{2n-1}$ entries of $\omega_{2n|0}$. Then we have
\begin{equation}
 \omega_{2n|0} = \begin{bmatrix}\omega_{2n-1|0}\\ \tilde{I}_{2n-1}\cdot\omega_{2n-1|0}\end{bmatrix}.
  \label{step2:4}   
\end{equation}
In order to derive the second entry of~(\ref{step2:4}) observe that the multiplication of $\omega_{2n-1|0}$ with~$\tilde{I}_{2n-1}$ turns~$\omega_{2n-1|0}$ upside down (i.e., the last entry of $\omega_{2n-1|0}$ becomes the first entry, the second last entry becomes the second entry and so on). Applying this multiplication to $\omega_{2n-1|0}$, which is written in the form of~(\ref{step2:3}), and using the fact that~$\omega_{2n-2|0}$ is a palindromic vector, we see that $\tilde{I}_{2n-1}\cdot\omega_{2n-1|0}$ is equal to the last~$2^{2n-1}$ entries, i.e., second half, of the vector~(\ref{proof:omega_rec}). By replacing $\omega_{2n|0}$ in~(\ref{step2:3}) with~(\ref{step2:4}), we obtain~(\ref{proof:omega:rec_odd}). The initial value $\omega_1 = \begin{bmatrix}0 & -2\end{bmatrix}^T$ follows directly by evaluating~(\ref{proof:omega_rec}) for $n=1$ and taking the first two entries.
\end{proof}
\begin{remark}
The recursions derived in Lemma~\ref{lemma:cond_entropy_rec} and~\ref{lemma:weighted_cond_entropy_rec} are with respect to initial state~$s_0=0$. They can be easily converted to recursions with respect to initial state~$s_0=1$ by using~(\ref{trafo1_h}) and~(\ref{trafo1_w}) from Lemma~\ref{lem:trafo}.
\end{remark}

\subsection{Proof of the Main Result}
By evaluating~(\ref{thm:ash:eq2}) based on Lemma~\ref{lemma:weighted_cond_entropy_rec}, we find exact solutions to the optimization problem~(\ref{objective})-(\ref{optconstraint2}).
\begin{theorem}
Consider the convex optimization problem~(\ref{objective}) to (\ref{optconstraint2}). The absolute maximum for input blocks of even length~$2n$ is
\begin{equation}
C^{\uparrow}_{2n} = \frac{1}{2}\log_2\left(\frac{5}{2}\right)
\label{main:eq1}
\end{equation}
for all $n \in \mathbb{N}$. For input blocks of odd length $2n-1$, the absolute maximum is
\begin{equation}
C^{\uparrow}_{2n - 1} = \frac{1}{2n-1}\left[\log_2\left(\frac{5}{4}\right) + (n-1)\cdot\log_2\left(\frac{5}{2}\right)\right], 
\label{main:eq2}
\end{equation}
where $n \in \mathbb{N}$.
\label{main}
\label{thm:main}
\end{theorem}
\begin{proof}
Without loss of generality, the initial state is assumed to be $s_0=0$. Recall~(\ref{C:ash:matrix}), which for input blocks of length~$2n+k$ reads as
\begin{equation}
C^{\uparrow}_{2n+k} = \frac{1}{2n+k}\log_2\left[1_{2n+k}^T \exp_2\left(\omega_{2n+k|0}\right)\right]
\label{main:C}
\end{equation}
where $n \in \mathbb{N}_0, k =1,2$. For $n = 0$, a straightforward computation shows that $C^{\uparrow}_{1} = \log_2\left(\frac{5}{4}\right)$ and  $C^{\uparrow}_{2} = \frac{1}{2}\log_2\left(\frac{5}{2}\right)$. Now assume that (\ref{main:eq1}) and (\ref{main:eq2}) hold for some~$n$. In particular, suppose
\begin{equation}
  1^T_{2n}\exp_2\left(\omega_{2n|0}\right) = \left(\frac{5}{2}\right)^n
  \label{ind_hyp1}
\end{equation}
and
\begin{equation}
1^T_{2n-1}\exp_2\left(\omega_{2n-1|0}\right) = \frac{5}{4}\left(\frac{5}{2}\right)^{n-1}.
  \label{ind_hyp2}
\end{equation}
Replacing~$\omega_{2n+2|0}$ and~$\omega_{2n+1|0}$ with the recursions derived in Lemma~\ref{lemma:weighted_cond_entropy_rec}, we obtain
\begin{align}
 1^T_{2n+2}\exp_2\left(\omega_{2n+2|0}\right) &= 1^T_{2n}\left[2\exp_2\left(\omega_{2n|0}\right) + 2\exp_2\left(\omega_{2n|0}-2\cdot 1_{2n}\right)\right]\nonumber\\
 &= \left(2 + 2\cdot2^{-2}\right) 1^T_{2n}\exp_2\left(\omega_{2n|0}\right)\nonumber
\end{align}
and
\begin{align}
 1^T_{2n+1}\exp_2\left(\omega_{2n+1|0}\right) &= 1^T_{2n-1}\left[2\exp_2\left(\omega_{2n-1|0}\right) + 2\exp_2\left(\omega_{2n-1|0}-2\cdot 1_{2n}\right)\right]\nonumber\\
 &= \left(2 + 2\cdot2^{-2}\right) 1^T_{2n-1}\exp_2\left(\omega_{2n-1|0}\right).\nonumber
 \end{align}
 Hence, using (\ref{main:C}) and the induction hypotheses (\ref{ind_hyp1}) and (\ref{ind_hyp2}), we have
 \begin{align}
 C^{\uparrow}_{2n+2} &= \frac{1}{2n+2}\log_2\left[\left(2 + 2\cdot2^{-2}\right) 1^T_{2n}\exp_2\left(\omega_{2n|0}\right)\right] \nonumber\\
 &= \frac{1}{2}\log_2\left(\frac{5}{2}\right)\nonumber
 \end{align}
 and
 \begin{align}
 C^{\uparrow}_{2n+1} &= \frac{1}{2n+1}\log_2 \left[\left(2 + 2\cdot2^{-2}\right) 1^T_{2n-1}\exp_2\left(\omega_{2n-1|0}\right)\right] \nonumber\\
 &= \frac{1}{2n+1}\left[\log_2\left(\frac{5}{4}\right) + n\cdot\log_2\left(\frac{5}{2}\right)\right].\nonumber
\end{align}
\end{proof}
\begin{remark}
Observe that $\lim_{n\rightarrow \infty}C^{\uparrow}_{2n+1} = \frac{1}{2}\log_2\left(\frac{5}{2}\right)$, where convergence is from below. Hence, we have
\[
\max_{n \in \mathbb{N}} C_n^{\uparrow} = \frac{1}{2}\log_2\left(\frac{5}{2}\right). 
\]
\end{remark}
Unfortunately, the distributions corresponding to~(\ref{main:eq1}) and~(\ref{main:eq2}) involve negative ``probabilities'' -- otherwise the capacity of the trapdoor channel would have been established. We state this as a formal remark.
\begin{remark}
Condition (\ref{thm:ash:condition}) does not hold for all $k = 1,\dots,2^n$, which can be seen as follows. For a trapdoor channel~$P_{n|0}$, we have
\begin{equation}
\begin{bmatrix}d_k\end{bmatrix}_{1\leq k \leq 2^n} = \left(P^{-1}_{n|0}\right)^T\exp_2\left(\omega_n\right). 
\label{rem:dk}
\end{equation}
Applying (\ref{invLemma1}) to $P_{n|0}$, which is written in the form of (\ref{Pn+1|0}), and taking the transpose, then applying~(\ref{invLemma1}) to the right bottom block of this matrix and taking the transpose and so on eventually shows that the second last row of $\left(P^{-1}_{n|0}\right)^T$ equals
\begin{equation}
\begin{bmatrix}0 & \cdots & 0 & 2^{n-1} & -2^{n-1}\end{bmatrix}. \nonumber
\end{equation}
Moreover, using Lemma~\ref{lemma:weighted_cond_entropy_rec}, it follows that the second to last entry and the last entry in~$\omega_n$ equals~$-2$ and~$0$, respectively. Inserting the gathered quantities into~(\ref{rem:dk}) yields
\begin{equation}
d_{2^n-1} = -3\cdot2^{n-3} < 0,\quad n \in \mathbb{N}. \nonumber
\end{equation}
\end{remark}
\section{Conclusions}
\label{sec:concl}
We have presented two different views on the trapdoor channel. The fractal view was motivated by the wish to find an explicit expression for the trapdoor channel -- a feature which would greatly simplify the capacity problem. Furthermore, the various views motivate using tools from other fields, e.g., fractal geometry.

Subsequently, we have focused on the convex optimization problem~(\ref{objective}) to~(\ref{optconstraint2}) where the feasible set is larger than the probability simplex. An absolute maximum of the $n$-letter mutual information was established for any~$n \in \mathbb{N}$ by using the method of Lagrange multipliers. The same absolute maximum  $\frac{1}{2}\log_2\left(\frac{5}{2}\right)\approx 0.6610$~b/u results for all even $n$ and the sequence of absolute maxima corresponding to odd block lengths converges from below to $\frac{1}{2}\log_2\left(\frac{5}{2}\right)$~b/u as the block length increases. Unfortunately, all absolute maxima are attained outside the probability simplex. Hence, instead of establishing the capacity of the trapdoor channel, we have shown only that~$\frac{1}{2}\log_2\left(\frac{5}{2}\right)$~b/u is an upper bound on the capacity. This upper bound is, to be best of our knowledge, the tightest known bound. Notably, this upper bound is strictly smaller than the feedback capacity~\cite{PeHaCuRo08}. Moreover, the result gives an indirect justification that the capacity of the trapdoor channel is attained on the boundary of the probability simplex.

\section*{Acknowledgment}
The author is supported by the German Ministry of Education and Research in the framework of the Alexander von Humboldt-Professorship and would like to thank Prof. Haim Permuter who suggested to use \cite[Theorem 3.3.3]{Ash65}. Moreover, the author wishes to thank Prof. Gerhard Kramer and Prof. Tsachy Weissman for helpful discussions. 
\bibliographystyle{IEEEtran}
\bibliography{techReport_v2}

\begin{thebibliography}{10}
\providecommand{\url}[1]{#1}
\csname url@samestyle\endcsname
\providecommand{\newblock}{\relax}
\providecommand{\bibinfo}[2]{#2}
\providecommand{\BIBentrySTDinterwordspacing}{\spaceskip=0pt\relax}
\providecommand{\BIBentryALTinterwordstretchfactor}{4}
\providecommand{\BIBentryALTinterwordspacing}{\spaceskip=\fontdimen2\font plus
\BIBentryALTinterwordstretchfactor\fontdimen3\font minus
  \fontdimen4\font\relax}
\providecommand{\BIBforeignlanguage}[2]{{%
\expandafter\ifx\csname l@#1\endcsname\relax
\typeout{** WARNING: IEEEtran.bst: No hyphenation pattern has been}%
\typeout{** loaded for the language `#1'. Using the pattern for}%
\typeout{** the default language instead.}%
\else
\language=\csname l@#1\endcsname
\fi
#2}}
\providecommand{\BIBdecl}{\relax}
\BIBdecl

\bibitem{Bla61}
D.~Blackwell, \emph{Information Theory}, E.~F. Beckenbach, Ed.\hskip 1em plus
  0.5em minus 0.4em\relax McGraw-Hill Book Co., New York, 1961, vol. Modern
  Mathematics for the Engineer.

\bibitem{Ash65}
R.~Ash, \emph{Information Theory}.\hskip 1em plus 0.5em minus 0.4em\relax
  Interscience Publishers, 1965.

\bibitem{KoMo02}
K.~Kobayashi and H.~Morita, ``An input/output recursion for the trapdoor
  channel,'' in \emph{Proc.\ IEEE Int.\ Symp.\ Inf.\ Theory}, Lausanne,
  Switzerland, Jun.~30--Jul.~5 2002, p. 423.

\bibitem{AhKa87}
R.~Ahlswede and A.~H. Kaspi, ``Optimal coding strategies for certain permuting
  channels.'' \emph{IEEE Trans.\ Inf.\ Theory}, vol.~33, no.~3, pp. 310--314,
  1987.

\bibitem{PeHaCuRo08}
H.~Permuter, P.~Cuff, B.~Van~Roy, and T.~Weissman, ``Capacity of the trapdoor
  channel with feedback,'' \emph{IEEE Trans.\ Inf.\ Theory}, vol.~54, no.~7,
  pp. 3150--3165, Jul. 2008.

\bibitem{Bar88}
M.~Barnsley, \emph{Fractals Everywhere}.\hskip 1em plus 0.5em minus 0.4em\relax
  Academic Press, Inc., 1988.

\bibitem{CovTho91}
T.~M. Cover and J.~A. Thomas, \emph{Elements of Information Theory},
  2nd~ed.\hskip 1em plus 0.5em minus 0.4em\relax John Wiley \& Sons, Inc.,
  1991.

\bibitem{Rob14}
E.~Roberts, \emph{Programming Abstractions in C++}.\hskip 1em plus 0.5em minus
  0.4em\relax Prentice Hall, 2014.

\bibitem{Gal68}
R.~G. Gallager, \emph{Information Theory and Reliable Communication}.\hskip 1em
  plus 0.5em minus 0.4em\relax John Wiley \& Sons, Inc., 1968.

\bibitem{GolLoa96}
G.~H. Golub and C.~F. van Van~Loan, \emph{Matrix Computations}, 3rd~ed.\hskip
  1em plus 0.5em minus 0.4em\relax The Johns Hopkins University Press, 1996.

\end{thebibliography}
\end{document}